\documentclass[10pt]{article} 

\usepackage{tikz}
\usetikzlibrary{shapes.geometric}
\usepackage{amsmath}
\usepackage{amsthm}
\usepackage{amsfonts}
\usepackage{amssymb}
\usepackage{array}
\usepackage[utf8]{inputenc}
\usepackage{geometry} 
\usepackage{graphicx} 
\usepackage{algorithm}
\usepackage{algorithmic}
\usepackage[T1]{fontenc}
\usepackage[utf8]{inputenc}
\usepackage{authblk}

\newtheorem{theorem}{Theorem}[section]

\newtheorem{proposition}[theorem]{Proposition}
\newtheorem{corollary}[theorem]{Corollary}
\newtheorem{lemma}[theorem]{Lemma}

\theoremstyle{definition}
\newtheorem{definition}[theorem]{Definition}
\newtheorem{example}[theorem]{Example}
\newtheorem{remark}[theorem]{Remark}

\newcommand{\C}{\mathcal{C}}

\newcommand{\F}{\mathbb{F}}

\newcommand{\N}{\mathbb{N}}
\newcommand{\Q}{\mathbb{Q}}

\newcommand{\bo}{\mathcal O}
\newcommand{\M}{\mathcal M}

\newcommand{\ch}{\mathrm{Char}}

\title{Maximum weight spectrum codes}
\author[1]{Tim Alderson \thanks {The author acknowledges the support of the NSERC of Canada Discovery Grant program.}}
\author[2]{Alessandro Neri \thanks {The author acknowledges the support of  Swiss National Science Foundation grant n.169510.}}
\affil[1]{University of New Brunswick, Canada}
\affil[2]{University of Z{u}rich, Switzerland}

\begin{document}
\maketitle
\begin{abstract}
%Given a linear $[n,k]_q$ code $\C$, 
%that is a $k$-dimensional subspace of $\mathbb F_q^n$, 
%we study the set of non-zero distinct weights of the elements in $\mathcal C$, called \emph{weight set}. Its notion can be traced in many coding theory papers, because of its  relation with the weight distribution of the code \cite{}.
%In this work we  consider linear $[n,k]_q$ codes, that are $k$-dimensional $\F_q$-subspaces of $\mathbb F_q^n$. The set of non-zero distinct weights of a code is called \emph{weight set}, and its notion can be traced in many coding theory papers, because of its clear connection with the weight distribution of the code \cite{}.

 In the recent work \cite{shi18}, a combinatorial problem concerning linear codes over a finite field $\F_q$ was introduced. In that work the authors studied the weight set of an $[n,k]_q$ linear code, that is the set of non-zero distinct Hamming weights, showing that its cardinality is bounded above by $\frac{q^k-1}{q-1}$.  They showed that this bound was sharp in the case $ q=2 $, and in the case $ k=2 $. They conjectured  that the bound is sharp for every prime power $ q $ and every positive integer $ k $.  In this work quickly establish the truth of this conjecture. We provide two  proofs, each employing different construction techniques. The first relies on the geometric view of linear codes as systems of projective points. The second approach is purely algebraic. We establish some lower bounds on the length of codes that satisfy the conjecture, and the length of the new codes constructed here are discussed.

\end{abstract}

\section{Introduction}
In 1973 Delsarte studied the number of distinct distances for a code $\C$. In the linear case, this reduces to studying the number of distinct weights  of the given code \cite{delsarte1973four}. In that work he underlines the importance of this parameter, analyzing its connections with the number of distinct weights of the dual code, and the minimum distance of the code and the minimum distance of the dual. These four parameters are studied in order to obtain various results on the distance properties, and, in particular, they are used to calculate the weight distributions of cosets of a code. 

 Discussions on set of distinct weights of a code  can be traced in \cite{mac63}, where the author skirmishes with the following question. Given a set of positive integers, $ S $,  is it possible to construct a code whose set of non-zero weights is $ S $? Partial solutions are presented, and  necessary conditions are established.

Recently, in \cite{shi18} the authors studied a combinatorial problem concerning the number of distinct weights 
%(weight spectrum) 
of linear codes. For a  code of dimension $k$ over the finite field $\F_q$, they showed that the size of the weight set  is bounded above by $\theta_q(k-1)=\frac{q^k-1}{q-1}$. They proved this bound to be sharp for binary codes, and for all $ q $-ary codes of dimension $ k=2 $.  They conjectured that the bound is sharp for all $ q $ and $ k $. Codes meeting this bound are called maximum weight spectrum (MWS) codes.

In this work we first quickly establish the existence of MWS codes for all $ k,q $. We provide two different constructions of $[n,k]_q$ MWS codes. In section 2 we give a brief recap on linear and projective codes, and define the basic tools needed for our constructions.
In section 3 we give a short proof of the existence of MWS codes via a geometric construction. The construction is pleasingly simple, but provides codes of  ``large'' length. 
A different approach is taken in Section 4. The  construction presented there is inductive, for dimension $ k\ge 1 $, and relies on algebraic tools.
In section 5 we investigate lower bounds on the length of MWS codes. We provide a geometric construction of a new infinite family of ``shorter" MWS codes, and we determine  the asymptotic length of the codes arising from both our algebraic, and our geometric construction.
Finally in Section 6 we summarize our work, and discuss some remaining questions.

\section{Preliminaries}

\subsection{Linear codes}

Let $q$ be a prime power and $\F_q$ denote the finite field with $q$ elements.  Recall that there always exists $\alpha \in \F_q$ that is a primitive element, i.e. $\F_q\setminus \{0\}= \F_q^*=\langle\alpha\rangle$. Throughout what follows, $\alpha$ shall denote a primitive element of $\F_q$.

Let $n$ be a positive integer. The \textit{Hamming distance} between two elements $a,b \in \F_q^n$ is defined as
$$d_H(a,b)=|\{i \in \{1,\ldots, n\} \mid a_i \neq b_i \}|.$$
It is well-known that the Hamming distance induces a metric on $\F_q^n$. The \textit{Hamming weight} of a codeword $c\in \F_q^n$ is defined as
$$w(c)=d_H(c,0)=|\{i \mid c_i \neq 0 \}|. $$

\begin{definition}
For an $[n,k]_q$ code $\C$ we define the \emph{weight set of $\C$} as
$$w(\C)=\left\{w(c) \mid c \in \C\setminus \{0\}\right\}.$$
%The \textit{weight spectrum} of $ \C $ is $ |w(\C)| $.
\end{definition}

%Throughout the whole work we will denote by $\theta_q(k)=\frac{q^{k+1}-1}{q-1}$, for $q$ prime power and $k$ positive integer.

Given two vectors $a \in \F_q^{n_1}$, $b \in \F_q^{n_2}$ we will use the notation $( a \mid b)$ to denote the vector in $\F_q^{n_1+n_2}$ obtained by concatenating  $a$ and $b$, i.e.
$$(a\mid b)=(a_1,\ldots, a_{n_1}, b_1, \ldots, b_{n_2}),$$
where $a=(a_1,\ldots, a_{n_1})$ and $b=( b_1, \ldots, b_{n_2}).$

\begin{definition}
Let $0<k \leq n$ be two positive integers. An $[n,k]_q$ \emph{code} $\C$  is a $k$-dimensional subspace of $\F_q^n$ equipped with the Hamming distance. A \emph{generator matrix} $G$ for $\C$ is a $k\times n$ matrix over $\F_q$ whose row vectors generate $\C$. The \emph{minimum distance} $d$ of $\C$ is the quantity $d=\min\{d_H(u,v)\mid u,v \in \C, u\neq v\}$.

An $[n,k]_q$ code  $ \C $ of dimension $ k\ge 2 $ is said to be   \emph{non-degenerate} if no coordinate position is identically zero.
\end{definition}

Throughout,  by $[n,k,d]_q$ code we will denote an $[n,k]_q$ code $\C$ whose minimum distance is $d$. Moreover, unless specified otherwise, all codes discussed here are assumed to be non-degenerate.
%an element $a \in \F_q^n$ will be always considered as a row vector, i.e. $a=(a_1,\ldots, a_n)$. 

\begin{definition}
 Let $\mathcal G$ be the subgroup of the group of linear automorphisms of $\F_q^n$ generated by the permutations of coordinates and by the multiplication of the $i$-th coordinate by elements in $\F_q^*$. Two codes $\C$ and $\C'$ are said to be \emph{equivalent} if there exists $\sigma \in \mathcal G$ such that $\C'=\sigma(\C)$.
\end{definition}

%An $[n,k]_q$ code  $ \C $ of dimension $ k\ge 2 $ is said to be  \emph{non-degenerate}) if no coordinate position is identically zero. Unless specified otherwise, all linear codes discussed here are assumed to be non-degenerate. 

\begin{definition}
Let $\beta \in \F_q$ and $c\in \F_q^n$. We define the number
$$c[\beta]=\left| \left\{ i \in \{1,\ldots,n\} \mid c_i=\beta  \right\}\right|.$$
and the \emph{entries distribution vector for $c$} as
$$V(c):=(c[\alpha], c[\alpha^2], \ldots, c[\alpha^{q-1}],c[0]) \in \N^{q}. $$
\end{definition}

Some basic properties concerning the entries distribution vector $V(c)$ are presented in the following. The proofs follow readily from the respective definitions. 

\begin{proposition}\label{prop:Vc}
 Let $c \in \F_q^n$, $\beta \in \F_q^*$ and let $e \in \F_q^n$ be the vector whose entries are all equal to $1$. The following hold: 
 \begin{enumerate}
 \item $V(\beta c)=(c[\frac{\alpha}{\beta}],c[\frac{\alpha^2}{\beta}], \ldots, c[\frac{\alpha^{q-1}}{\beta}],c[0])$. In particular, since $\beta=\alpha^{j}$ for some $j$, then the vector consisting of the first $q-1$ entries of $V(\beta c)$ is the $j$-th shift of the vector formed by the first $q-1$ entries of $V(c)$.
 \item $V(c+\beta e)=(c[\alpha-\beta],c[\alpha^2-\beta], \ldots, c[\alpha^{q-1}-\beta], c[-\beta]) $ and therefore $V(c+\beta e)$ is a permutation of the vector $V(c)$.
 \item If $c=(a \mid b)$, then for every $\beta \in \F_q$, $c[\beta]=a[\beta]+b[\beta]$, i.e. $V(c)= V(a)+V(b)$. 
 \item $c[\beta]=n-w(c-\beta e)$ and $c[0]=n-w(c)$.
 \item $\langle V(c),e\rangle=n$, where $\langle\cdot,\cdot\rangle$ denotes the standard inner product.
 \end{enumerate} 
\end{proposition}

%\begin{proof}
%\begin{enumerate}
% \item It directly follows from the fact that $(\beta c)_\ell=\alpha^i$ if and only if $c_\ell=\frac{\alpha^i}{\beta}$.
%\item $(c+\beta e)_\ell= \alpha^i$ if and only if $c_\ell=\alpha^i-\beta$.
%\item Trivial.
%\item $n=|\{i \mid c_i=0\}| + |\{ i \mid c_i \neq 0\}|=c[0]+w(c)$.
%\item $<V(c),e>=\sum\limits_{\beta \in \F_q^*}c[\beta]+ c[0]=|\{i \mid c_i=0\}| + |\{ i \mid c_i \neq 0\}|=n$.
%\end{enumerate}
%\end{proof}

\begin{definition}
Given an $[n,k]_q$ code $\C$ with generator matrix $G$ and $r=(r_1,\ldots, r_{q-1})\in \N^{q-1}$, we define the \emph{generalized $r$-repetition code of $\C$}
as the code $\C(r)$ whose generator matrix is 

\[
[\underbrace{\alpha G \mid \ldots  \mid
            \alpha G}_{r_1 \text{ times}}\mid \underbrace{\alpha^2 G \mid \ldots \mid \alpha^2 G}_{ r_2 \text{ times}} \mid \ldots \ldots \mid \underbrace{\alpha^{q-1} G \mid \ldots \mid \alpha^{q-1} G}_{ r_{q-1} \text{ times}} ]
\]
i.e.
$$\C(r)=\{c^r:=(\underbrace{\alpha c \mid \ldots  \mid
            \alpha c}_{r_1 \text{ times}}\mid \underbrace{\alpha^2 c \mid \ldots \mid \alpha^2 c}_{ r_2 \text{ times}} \mid \ldots \ldots \mid \underbrace{\alpha^{q-1} c \mid \ldots \mid \alpha^{q-1} c}_{ r_{q-1} \text{ times}}) \mid c \in \C  \}$$
\end{definition}            
            
The next result explains some properties of the code $\C(r)$.

\begin{proposition}\label{prop:Cr}
Let $\C(r)$ be the generalized $r$-repetition code of an $[n,k]_q$ code $\C$ for a non-zero vector $r=(r_1,\ldots, r_{q-1})\in \N^{q-1}$. Let moreover $R:=r_1+\ldots+r_{q-1}$. The following hold:
\begin{enumerate}
\item $\C(r)$ is an $[Rn,k]$ linear code over $\F_q$.
\item $w(c^r)=Rw(c)$. 
\item $c^r[0]=Rc[0]$.
\item $|w(\C)|=|w(\C(r))|$. In particular, if $\C$ is an MWS code, then also $\C(r)$ is an MWS code.
\item For every $i=1,\ldots,q-1$, 
$$c^r[\alpha^i]=\sum_{j=1}^{q-1}r_jc[\alpha^{i-j}].$$
\end{enumerate}

\end{proposition}            
            
\begin{proof}
 \begin{itemize}
 \item[1-4.] Follow from the definition.
% \item For every $i=1,\ldots, q-1$, $ \alpha^i c_j \neq 0$ if and only if $c_j \neq 0$.
% \item Same argument as in 2.
% \item It is an immediate consequence of 2.
\item[5.] It is an easy calculation, that follows from part 1 and 3 of Proposition \ref{prop:Vc}.
\end{itemize}
\end{proof} 

\subsection{Projective systems}

In this section we introduce the geometric view of linear codes, as detailed in \cite{MR1186841} (or in \cite{AB3} for codes that are equivalent to linear).  We start with a short
overview of the fundamentals of finite projective geometry. For a detailed introduction we refer to the recent book by Ball \cite{Ball2015}. We let $PG(k,q)$ represent the finite
projective geometry of dimension $k$ and order $q$.  Due to a result of Veblen and Young
\cite{MR0179666}, all finite projective spaces of dimension greater than two are isomorphic up to
the order $q$. The space $PG(k,q)$ can be modelled most easily with the vector space of dimension
$k+1$ over the finite field $\F_q$.  In this model, the one-dimensional subspaces represent the
points, two-dimensional subspaces represent lines, etc. Formally, we have

$$ PG(k,q):= \left(\F_q^{k+1}\setminus \{0\}\right)/_\sim, $$
where
$$u\sim v \mbox{ if and only if } u=\lambda v \mbox{ for some } \lambda \in \F_q.$$

 Using this model, it is not hard to show
by elementary counting that the number of points of $PG(k,q)$ is given by \[\theta_q(k)=\frac{q^{k+1}-1}{q-1}.\] 
%Throughout the whole work $\theta_q(k)$ will always denote that quantity.

A \textit{$d$-flat} $\Pi$ in $PG(k,q)$ is a subspace isomorphic to $PG(d,q)$; if $d=k-1$, the subspace $\Pi$ is called a \textit{hyperplane}. 
%A property that we shall find useful is the principle of duality. For any
%space $S = PG(k,q)$, there is a \emph{dual} space $S^*$ whose points and hyperplanes are respectively the hyperplanes and points of $S$.  For any result about points of $S$, there is always a corresponding result about hyperplanes of $S^*$. More generally, for any result dealing with subspaces of $S$, replacing each reference to a subspace $PG(m,q)$, $m < k$, with a reference to the subspace $PG(k-m-1,q)$ yields a correspond \emph{dual} statement of $S^*$ that has the same truth value. For instance, a result about a set of points of $PG(k,q)$, no three of which are collinear, could be rewritten dually about a set of hyperplanes of $PG(k,q)$, no three of which meet in a common subspace of dimension $k-2$.
 Central to the geometric view of linear codes is the idea of a projective system. 

\begin{definition}
	A \emph{projective $[n,k,d]_q$-system} is a finite (multi)set $\M$ of points 
	of  $PG(k-1,q)$, not all of which lie in a hyperplane, where $n=|\M|$ , and $$d=n-\max\{ |\M\cap H| \mid H \subset PG(k-1,q), \dim(H)=k-2\}.$$ Note that the cardinalities above are counted with multiplicities in the case of a multiset. We denote by $ m(P) $ the multiplicity of the point $ P $ in $ \M $.  \\
	Two  \emph{projective $[n,k,d]$-system} $\M$ and $\M'$ are said to be \emph{equivalent} if there exists a projective isomorphism of $PG(k-1,q)$ mapping $\M$ to $\M'$.
\end{definition}

Let $ \C $  be an $ [n,k]_q $ code with $ k\times n $ generator matrix $ G $. Note that multiplying any column of $ G $ by a non-zero field element yields a generator matrix for a code which is equivalent to $ \C $. Consider the (multi)set of one-dimensional subspaces of $ \F_q^n $ spanned by the columns of $ G $. In this way the columns may be considered as a (multi)set $ \M $ of points of $ PG(k-1,q) $.  

For any non-zero vector $ v=(v_1,v_2,\ldots,v_k) $ in $ \F_q^k $, it follows that the projective hyperplane 
\[
v_1x_1+v_2x_2+\cdots + v_kx_k=0
\]
 contains $ |\M|-w $ points of $ \M $ if and only if the codeword $ vG $ has weight $ w $.  It follows that  linear (non-degenerate)  $ [n,k,d]_q $ codes and projective $ [n,k,d]_q $  systems are equivalent objects. That is to say, there exists a linear $ [n,k,d]_q $ code if and only if there exits a projective $ [n,k,d]_q $ system.

%
%Consider now a linear $[n,k,d]_q$ non-degenerate code $\C$ over $\F_q$, and its generator matrix $G$ of the form
%$$G=\left( v_1 \mid v_2 \mid \ldots \mid v_n\right)$$
%Then we can consider the projective system in $PG(k-1,q)$ whose points are given by the columns $v_i$'s of $G$. This projective system has the same parameter as the code, i.e.  it is a projective $[n,k,d]_q$-system, as the following result shows.
%
%\begin{theorem}\cite[Theorem 1.1.6]{MR1186841}
% Let $k\ge 1$ and $d\ge1$ be two positive integers. There is a 1-to-1 correspondence between the set of equivalence 
%classes of non-degenerate linear $[n,k,d]_q$ codes and the set of equivalence classes of projective $[n,k,d]_q$-systems. 
%\end{theorem}
%In fact there is a one-to-one correspondence between (equivalence classes of) $[n,k,d]$ non-degenerate codes over $\F_q$ and (equivalence classes of) projective $[n,k,d]$-systems (see \cite[Theorem 1.1.6]{MR1186841} and in general Section 1.1.2). 

\subsection{Maximum weight spectrum codes}

The weight set of a code has been studied in many contexts of coding theory, and for different purposes. In \cite{mac63}, using the relation between the weight distributions of a code and its dual, the author investigated necessary conditions for the existence of a linear binary code with a given weight set. One of the first to study the cardinality of the weight set of a code was Delsarte \cite{delsarte1973four}. He demonstrated its importance in computing the weight distributions of cosets of a code. 
Other problems concerning the weight set and its cardinality can be found in \cite{sl56, en87}.

Recently in \cite{shi18}, Shi \textit{et. al.}  investigated the maximum cardinality of the weight set of a code, showing the following upper bound.

\begin{proposition}\cite[Proposition 2]{shi18}\label{prop:boundMWS}
 If $\C$ is an $[n,k]_q$ code, then 
$$ |w(\C)| \leq \theta_q(k-1),$$
\end{proposition}

Motivated by Proposition \ref{prop:boundMWS}, we define a new family of codes.

\begin{definition}
 An $[n,k]_q$ code $\C$ such that $|w(\C)|=\theta_q(k-1)$ is called a \emph{maximum weight spectrum (MWS) code}.
\end{definition}

\begin{remark}
Observe that this definition is coherent with the existing literature. If  $\C$ is an $[n,k]_q$ code, then the \emph{weight spectrum} of $\C$ typically denotes the vector $A(\C)=(A_1,\ldots, A_n)$, where 
$$A_i=\left|\left\{c\in \C \mid w(c)=i \right\}\right|.$$
In this framework, the cardinality of the weight set of $\C$ coincides with the Hamming weight of the vector $A(\C)$, and  $\C$ is MWS if and only if the Hamming weight of $A(\C)$  is $\theta_q(k-1)$. Since every non-zero entry of $A(\C)$ is a multiple of $q-1$, $\theta_q(k-1)$ is actually the maximum possible value for the Hamming weight of $A(\C)$.
\end{remark}

In \cite{shi18} the authors conjectured, motivated by experimental results, that for every $q$ and $k$, MWS codes exist.In the following section, we quickly establish the truth of this conjecture.

\section{A geometric construction of MWS codes}

In this section we are going to give a geometric construction of $ [n,k]_q $ MWS codes for every prime power $q$ and every $k\ge 2$. 

Given an $[n,k,d]_q$ code $\C$, we can consider the associated projective $[n,k,d]_q$-system $\M(\C)$, whose points are given by the columns of the generator matrix. 
%By duality, we can also consider the corresponding family  $\M(\C)^*$ of hyperplanes in $\Pi:=PG(k-1,q)$. 

%\begin{definition}
%Let $\M$ be a projective system in $ \Pi=PG(k-1,q) $. We define the \emph{character} function of $ \M $, denoted $ \ch_{\M} $, mapping the hyperplanes of $ \Pi $ to the non-negative integers:
% \[
% \ch_\M(H)=\sum_{P\in H} m(P).
% \]
%So $\ch_{\M}(H) $ is  the number, including multiplicity, of points in $ \M \cap H$. 
%Dually, we define the character of a point $ P $  with respect to $ \M^*$, denoted $\ch_{\M^*}(P)$, to be the number, including multiplicity, of hyperplanes in $ \M^*$ that are incident with $ P $.
%\end{definition}

\begin{definition}
Let $\M$ be a multiset in $ \Pi=PG(k-1,q) $. We define the \emph{character} function of $ \M $, denoted $ \ch_{\M} $, mapping the power set of $ \Pi $ to the non-negative integers:
 \[
 \ch_\M(A)=\sum_{P\in A} m(P).
 \]
So $\ch_{\M}(A) $ is  the number, including multiplicity, of points in $ \M \cap A$. With a slight abuse of notation, we will write $ m(P)=\ch_{\M}(P) $, for any point $ P $.
%Dually, we define the character of a point $ P $  with respect to $ \M^*$, denoted $\ch_{\M^*}(P)$, to be the number, including multiplicity, of hyperplanes in $ \M^*$ that are incident with $ P $.
\end{definition}

The following follows directly from the definitions.

%\begin{lemma}
%Let $\C$ be an $[n,k]_q$ code over $\F_q$, and let $\M:=\M(\C)$ be its associated projective system of hyperplanes. There exists a codeword of weight $ s $ in $\C $ if and only if there exists a point  $P\in  \Pi $ with $\ch_\M(P)= n-s $.
%\end{lemma}

\begin{lemma}
Let $\C$ be an $[n,k]_q$ code over $\F_q$, and let $\M:=\M(\C)$ be its associated projective system. There exists a codeword of weight $ s $ in $\C $ if and only if there exists a hyperplane  $H$ in  $ \Pi $ with $\ch_\M(H)= n-s $.
\end{lemma}

A natural consequence of the above Lemma is the following result on the existence of MWS codes.

%\begin{lemma}\label{lem:GeometricConditions}
%There exists an $[n,k]_q$ MWS code $\C $  if and only if there exists a projective system of hyperplanes $\M$ in $ \Pi= PG(k-1,q) $ such that for any two points $ P\ne Q $ in $ \Pi $,  $\ch_\M(P) \ne $ $\ch_\M(Q)$ 
%\end{lemma}

\begin{lemma}\label{lem:GeometricConditions}
There exists an $[n,k]_q$ MWS code $\C $  if and only if there exists an $ [n,k,d]_q $ projective system  $\M$  such that $ \ch_{\M} $ is injective.   
\end{lemma}

We now provide a construction of a projective system as required in the Lemma \ref{lem:GeometricConditions}.
Let $ \Pi=PG(k-1,q) $, $ k\ge 2 $ and let the points of $ \Pi$ be denoted $ P_0,P_1,\ldots,P_{\theta_q(k-1)-1} $.  For each $ i $, include $ P_i $ in $ \M $  with multiplicity $ 2^i $.\\

%Let $ P\in \Pi $ (a point), by the character of $ P $, denoted Char$ (P) $ we denote the number of hyperplanes in the multiset incident with $ P $.\\

For $ k=2 $, $ \Pi $  is the projective line, so clearly no two points will have the same character. Consider $ k\ge 3 $. Each hyperplane in $ \Pi $ is incident with precisely $ \theta_q{(k-2)} $ distinct points, and simple counting shows that every pair of distinct hyperplanes are incident with precisely $ \theta_q{(k-3)} $ distinct points. For any particular hyperplane $ H $, let us suppose that $ H $ is incident with $ P_{i_1}, P_{i_2}, \ldots,P_{i_{\theta_q(k-2)}}  $. It follows that 
\[
\ch_\M(H) = \sum\limits_{j=1}^{\theta_q(k-2)}2^{i_j}.
\]  
It follows that  no two hyperplanes have the same character (consider the binary expansion of the respective characters). We have therefore proved the following result.

\begin{theorem} \label{thm:LongGeomExamples}
For each prime power $q$ and  $ k\ge 2 $, there exists an  $[n,k]_q$ MWS code, where $ n=2^{\theta_q(k-1)}-1 $. 
%In particular,  $ L(k,q)=\theta_q(k-1) $, $ k\ge 2 $.
\end{theorem}
We note that the construction used in establishing the Theorem \ref{thm:LongGeomExamples} involves codes of considerable length (asymptotically).  A natural question is whether ``short''  MWS codes exist. We investigate this question in the sequel.

\section{An algebraic construction of MWS codes}\label{sec:alg}

In this section we give a different construction of MWS codes that relies on algebraic properties of linear codes. This construction is inductive, where the inductive step is divided in two parts. 

We now define two properties playing a central role in the next construction. 
\begin{equation}\label{propertyA} 
\text{\parbox{.90\textwidth}{ There exists  $ \beta \in \F_q^* $  such that, for $a,b\in \C$, $ a[\beta] = b[\beta] $ only if $ a=b$. }} \tag{A}
\end{equation}
\begin{equation} \label{propertyB}
\text{\parbox{.90\textwidth}{  $V(c)$  has pairwise distinct entries for every $c\in \C\setminus\{0\}$.}} \tag{B}
\end{equation}

\begin{proposition}\label{prop:recursive}
Let $q$ be a prime power, and let $\C$ be an $[n,k]_q$ MWS code. 
\begin{enumerate}
\item  \label{prop:recursive:pt1} If $ \C $  satisfies (\ref{propertyA}) then  then there exists an $[N,k+1]_q$ MWS code $\bar{\C}$, where $N=2n+1$.
\item  \label{prop:recursive:pt2}  Let  $q\geq 3$.  If  $\C$  satisfies both  (\ref{propertyA}) and  (\ref{propertyB}), then there exists an $[N,k+1]_q$ MWS code $\tilde{\C}$ which satisfies property  (\ref{propertyB}), where $N=(q-1)n+ (q-2)+(T+1)\frac{(q-2)(q-3)}{2}$
and $T=\max\{c[\beta] \mid c\in \C\setminus \{0\}, \beta \in \F_q^*\}$.
\end{enumerate}
\end{proposition}

\begin{proof}
%\begin{enumerate}
%\item 
Let $t=\frac{q^k-1}{q-1}$ and let $1\leq w_1 < w_2 <\ldots < w_t \leq n$ be the distinct weights of the code $\C$. %We construct the code $\tilde{C}$ as follows.
For $ N>0 $ consider the embedding 
 $$\begin{array}{rcl}
\phi :\C & \longrightarrow &\F_q^N \\
 c & \longmapsto & (c \mid 0 \ldots 0).
 \end{array}$$ 
For part 1 take $ N=2n+1 $ and let $\bar{\C}$ be the $[N,k+1]_q$ code generated by $\phi(\C)$ and $ e $, where $e$ is the vector whose entries are all equal to $1$. 
%$q^k$ additional distinct weights. In fact, 
For every $c\in \C$ we have $w(\phi(c))=w(c)$, and $w(\beta e-c)=N-c[\beta]$.  
Since $ w(c)\le n <N-n\le N-c[\beta]$, property  (\ref{propertyA})  gives part 1.\\ 

For part 2 take $ N=(q-1)n+ (q-2)+(T+1)\frac{(q-2)(q-3)}{2} $, and take  the $[N,k+1]_q$ code $ \tilde{\C}$  generated by $\phi(\C)$ and $ x $, where $x$ is the vector defined as
$$x=(\underbrace{1,\ldots,1}_{n \text{ times }},\underbrace{\alpha,\ldots,\alpha}_{n+1 \text{ times }},\underbrace{\alpha^2,\ldots,\alpha^2}_{n+1 +(T+1) \text{ times }},\ldots,\underbrace{\alpha^{q-2},\ldots,\alpha^{q-2}}_{n+1+(q-3)(T+1) \text{ times }})$$
The proof that this code is MWS is analogous to part 1, so we shall show that  (\ref{propertyB}) is satisfied. For $c\in\C$,  
 $$V(\phi(c))=V(c)+(\underbrace{0, \dots, 0}_{q-1 \text{ times}}, N-n ).$$
 $ \C $ satisfies (b), and each entry in $V(c)$ is strictly less than $N-n$ ($q\geq 3$). It follows that the entries of $V(\phi(c))$ are pairwise distinct. Now, for some $\alpha^j \in \F_q^*$, and $c\in \C$ we consider the  entries of $V(z)$, where $z=\alpha^j x+\phi(c)$. We notice that
 $$\phi(c)+\alpha^j x=(c+\alpha^j e \mid \underbrace{\alpha^{j+1},\ldots,  
            \alpha^{j+1}}_{n+1 \text{ times}}\mid \ldots \mid \underbrace{\alpha^{j-1},\ldots,
            \alpha^{j-1}}_{n+1+(q-3)(T+1)  \text{ times}}) $$
Therefore, by part 3 of Proposition \ref{prop:Vc}, $V(z)=V(c+\alpha^j e)+V(y)$,
 where $$y[\alpha^j]=y[0]=0, y[\alpha^{j+i}]=n+1+(i-1)(T+1).$$
 
 By part 2 of Proposition \ref{prop:Vc}, $V(c+\alpha^j e)$ is  a permutation of the vector $V(c)$ and hence the entries are pairwise distinct. This gives $z[\alpha^{j+i}]=c[\alpha^{j+i}]+n+1+(i-1)(T+1)$, for $i=1,\ldots,q-2$, while $z[0]=c[-\alpha^j]$ and $z[\alpha^j]=c[0]$. The result follows, since $c[\alpha^{j+i}]<T+1$ for $i=1,\ldots,q-2$, and $c[-\alpha^j]\neq c[0]$.
%\end{enumerate}
\end{proof}

%\begin{remark}
% The construction given in part 2 of the above Proposition works only for $q\geq 3$. Fixing some details it is possible to generalize it such that it works also in the case $q=2$. However, since that case has already been proved in \cite[Theorem 1]{shi18}, we omit it.
%\end{remark}

We want to use this result for an inductive construction of $[n,k]_q$ MWS codes.
However, starting with an $[n,k]_q$ MWS code, this construction gives a new $[N,k+1]_q$ MWS code $\tilde{\C}$ that does not satisfy  (\ref{propertyA}). In fact, for every $j=1,\ldots, q-1$, if we take $z_1=\alpha^j x+\phi(c)$ and $z_2=\alpha^j x+\phi(\lambda c)$ as in the proof above, for some $1\neq\lambda \in \F_q^*$, then $z_1\neq z_2$, and 
$z_1[\alpha^j]=c[0]=(\lambda c)[0]=z_2[\alpha^j]$. 

Starting from this code $\tilde{\C}$, we need to construct another MWS code that satisfies  (\ref{propertyA}). This can be done using the generalized $r$-repetition code of $\C$ with a suitable vector $r$, as we will see in the following. We first need an auxiliary result.

\begin{lemma}\label{lem:hyp}
 Suppose $H \subseteq \Q^m$ is a finite union of affine hyperplanes. Then there exist a non-zero vector $z=(z_1,\ldots, z_m) \in \N^m$ such that $z \notin H$.
\end{lemma}

\begin{proof}
 Let $s \in \N$ and $H=\cup_{j=1}^s H_j$, where 
 $$H_j=\left\{ x \in \Q^m \mid \sum_{i=1}^m f_i^{(j)} x_i = 0 \right\}, $$
with $f_i^{(j)}$'s not all zeros.
 We take a vector of the form $v=(1,t, \ldots, t^{m-1})$ and show that it cannot be in $H$ for every $t\in \N$. In fact, $v \in H$ if and only if there exist an $\ell$ such that $f^{(\ell)}(t):=\sum_{i=1}^m f_i^{(\ell)}t^{i-1}  = 0$, if and only if 
 $$F(t):=\prod_{j=1}^s f^{(j)}(t)=0.$$
 Since $F(T)$ is a non-zero polynomial in $\Q[T]$ of degree at most $s(m-1)$, it can not vanish on the whole $\N$.
\end{proof}

We now wish to show $a^r[\alpha^i]\neq b^r[\alpha^i]$ for every $a,b \in \C$ such that $a\neq b$, for some $i=1,\ldots, q-1$. Using part 5 of Proposition \ref{prop:Cr}, this equates to showing 
$$\sum_{j=1}^{q-1}r_ja[\alpha^{i-j}]\neq \sum_{j=1}^{q-1}r_jb[\alpha^{i-j}],$$
or, equivalently,
$$\sum_{j=1}^{q-1}r_j\left(a[\alpha^{i-j}]-b[\alpha^{i-j}]\right) \neq 0.$$
This is equivalent to the condition

\begin{equation}\label{eq:inters}r=(r_1,\ldots, r_{q-1}) \notin \bigcup_{\substack{a,b \in \C \\ a\neq b}} H_i^{a,b},\tag{$\star$} \end{equation}

where
$$H_i^{a,b}:=\left(a[\alpha^{i-1}]-b[\alpha^{i-1}], \ldots,a[\alpha^{i}]-b[\alpha^{i}]\right)^{\perp}.$$

\begin{remark}\label{rem:all}
 Observe that here the choice of $i$ does not really matter. In fact, if for all $a,b\in \C$ with $a\neq b$ we have that $a[\alpha^i]\neq b[\alpha^i]$ for some $i$, then $a[\alpha^j]\neq b[\alpha^j]$ for every $j=1,\ldots, q-1$. By contradiction, if $a[\alpha^j]= b[\alpha^j]$ for some $j$ and some $a,b\in \C$ with $a\neq b$, then we would have
$$a[\alpha^i]=(\alpha^{j-i}a)[\alpha^j]=(\alpha^{j-i}b)[\alpha^j]=b[\alpha^i],$$
which is not possible.
\end{remark}

The following Lemma gives necessary and sufficient conditions for the existence of a vector  $r=(r_1,\ldots, r_{q-1}) \in \N^{q-1}$ that satisfies (\ref{eq:inters}).

\begin{lemma}\label{lem:exist}
 A vector $r=(r_1,\ldots, r_{q-1})$ that satisfies  (\ref{eq:inters}) exists if and only if for every $a,b \in \C$ with $a\neq b$ we have
$$V(a)\neq V(b).$$
 \end{lemma}
 
 \begin{proof}
  Suppose $V(a)=V(b)$ for some $a\neq b$. Then the vector $\left(a[\alpha^{i-1}]-b[\alpha^{i-1}], \ldots,a[\alpha^{i}]-b[\alpha^{i}]\right)$ is the zero vector. This implies that
$$H_i^{a,b}=\left(a[\alpha^{i-1}]-b[\alpha^{i-1}], \ldots,a[\alpha^{i}]-b[\alpha^{i}]\right)^{\perp}=\Q^{q-1}$$
 and in (\ref{eq:inters}), no such $ r $ exists.

On the other hand, if $V(a)\neq V(b)$, then $V(a)$ and $V(b)$ differ in at least two entries. Hence there exists at least one $1\le i\le q-1$ such that $a[\alpha^i]-\beta[\alpha^i]\neq 0$, and the vector $(a[\alpha^{i-1}]-b[\alpha^{i-1}], \ldots,a[\alpha^{i}]-b[\alpha^{i}])$ is non-zero. 
%Since it holds for every $a,b \in \C$, with $a\neq b$, we have that
It follows that
$$ \bigcup_{\substack{a,b \in \C \\ a\neq b}} H_i^{a,b}$$
is a finite union of hyperplanes. By Lemma \ref{lem:hyp} there exists a solution for (\ref{eq:inters}) in $\N^{q-1}$.
 \end{proof}

Therefore, we want to impose the condition that $V(a)\neq V(b)$ for $a \neq b$.

\begin{lemma}\label{lem:lambda}
Let $\C$ be an $[n,k]_q$ MWS code. If $a, b \in \C\setminus \{0\}$ are such that $V(a)=V(b)$, then $a=\lambda b$ for some $\lambda \in \F_q^*$.
\end{lemma}

\begin{proof}
For $a, b \in \C\setminus\{0\}$, if $a\neq \lambda b$ for any $\lambda \in \F_q^*$, then, since $\C$ is a MWS code, $w(a)\neq w(b)$, and this implies $V(a)\neq V(b)$.
\end{proof}

\begin{corollary}\label{cor:exist1}
Let $\C$ be an $[n,k]_q$ MWS code that satisfies  (\ref{propertyB}). Then  (\ref{eq:inters}) has solution.
\end{corollary}

\begin{proof}
This means that $V(c)=V(\lambda c)$ if and only if $\lambda =1$. Moreover by Lemma \ref{lem:lambda} this implies $V(a)\neq V(b)$ for all distinct $a,b \in \C\setminus \{0\}$. Then we conclude by Lemma \ref{lem:exist}.
\end{proof} 

Now, we also want that property  (\ref{propertyB})  is preserved when we extend the code to the generalized $r$-repetition code $\C(r)$. This means, for every $c^r\in \C(r)\setminus\{0\}$,
$c^r[\alpha^i]\neq c^r[\alpha^\ell]$
for every $i<\ell$, and moreover $c^r[\alpha^\ell]\neq c^r[0]$ i.e.
%$$ \sum_{j=1}^{q-1}r_jc(\alpha^{i-j})\neq \sum_{j=1}^{q-1}r_jc(\alpha^{\ell-j})$$
$$ \sum_{j=1}^{q-1}r_j\left(c[\alpha^{i-j}]-c[\alpha^{\ell-j}]\right) \neq 0 \qquad \mbox{ and } \qquad \sum_{j=1}^{q-1}r_j\left(c[\alpha^{\ell-j}]-c[0]\right) \neq 0. $$
This condition can be reformulated as
$$(r_1,\ldots, r_{q-1}) \notin \bigcup_{i=0}^{q-1}\bigcup_{\ell=i+1}^{q-1} H_{i,\ell}^c, $$
where 

$$H_{i,\ell}^c:=\left(c[\alpha^{i-1}]-c[\alpha^{\ell-1}], \ldots,c[\alpha^{i}]-c[\alpha^{\ell}]\right)^{\perp}$$
for $1\leq i < \ell \leq q-1$, and
$$H_{0,\ell}^c:=\left(c[\alpha^{\ell-1}]-c[0], \ldots,c[\alpha^{\ell}]-c[0]\right)^{\perp}.$$
If we do this for every $c\in \C\setminus\{0\}$,  we get
\begin{equation}\label{eq:2}
r=(r_1,\ldots, r_{q-1}) \notin \bigcup_{c\in \C\setminus\{0\}} \bigcup_{i=0}^{q-1}\bigcup_{\ell=i+1}^{q-1} H_{i,\ell}^c.\tag{$\star\star$}
\end{equation}
As we did for  (\ref{eq:inters}), we now find conditions such that  (\ref{eq:2}) has solutions.
\begin{lemma}\label{lem:exist2}
 If  (\ref{propertyB}) holds, then  there exists a vector $ r $ satisfying (\ref{eq:2}).
\end{lemma}

\begin{proof}
 If $V(c)$ has all distinct elements it is clear that $H_{i,\ell}^c$ is an hyperplane for every $i<\ell$. Therefore, we conclude using Lemma \ref{lem:hyp}.
\end{proof}

The following theorem summarizes this part on the generalized $r$-repetition code $\C(r)$, fundamental tool for this construction of MWS codes.

\begin{theorem}\label{thm:cr}
Suppose $\C$  is an $[n,k]_q$ MWS code that satisfies property  (\ref{propertyB}). Then there exists $r=(r_1,\ldots,r_{q-1})\in \N^{q-1}$ such that $\C(r)$ is an $[Rn,k]_q$ MWS code that satisfies (\ref{propertyA}) and (\ref{propertyB}).
%\begin{enumerate}
% \item for every  $a^r,b^r \in \C(r)$ and for every $i=1,\ldots, q-1$, $a^r[\alpha^i]\neq b^r[\alpha^i]$,
%\item $\C(r)$ is an $[Rn,k]_q$ MWS code,
%\item for every $c^r \in  \C(r)$, $V(c^r)$ has pairwise distinct entries.
%\end{enumerate}
\end{theorem}

\begin{proof}
Consider $r=(r_1,\ldots,r_{q-1})\in \N^{q-1}$ not identically zero  such that both (\ref{eq:inters}) and (\ref{eq:2}) are satisfied. Such a vector $r$ exists, since we can always find, by Lemmas \ref{lem:hyp} and \ref{lem:exist2} and Corollary \ref{cor:exist1}, a non-zero vector in $\N^{q-1}$ that is not in
$$X_1\cup X_2,$$
where 
$$X_1= \bigcup_{\substack{a,b \in \C \\ a\neq b}} H_i^{a,b} \qquad \mbox{ and } \qquad X_2=\bigcup_{c\in \C\setminus\{0\}}\bigcup_{i=0}^{q-1}\bigcup_{\ell=i+1}^{q-1} H_{i,\ell}^c.$$

With this choice of $r$ the code $\C(r)$  satisfies (\ref{propertyA}) and (\ref{propertyB}) by construction (also see Remark \ref{rem:all}). Moreover, by part 4 of Proposition \ref{prop:Cr}, it is an $[Rn,k]_q$ MWS code.
\end{proof}

Now we are ready to give a second proof of the conjecture.

\begin{theorem}\label{thm:algexist} Let $q\geq 3$ be a prime power and $k$ be a positive integer. Then an $[n,k]_q$ MWS code  exists.  
%$$L(k,q) = \frac{q^k-1}{q-1}. $$
\end{theorem}

\begin{proof}
We prove by induction that there exists an $[n^{(k)},k]_q$ MWS code $\C_k$  that satisfies (\ref{propertyB}).

For the case $k=1$ we take the code $\C_1 \subseteq \F_q^{n}$ generated by the codeword
$$c=(1,\alpha, \alpha, \alpha^2, \alpha^2, \alpha^2, \ldots, \underbrace{\alpha^{q-2},\ldots,\alpha^{q-2}}_{q-1 \text{ times }}), $$
where $n^{(1)}=\frac{q(q-1)}{2}$. It is easy to see that $\C_1$ is MWS and satisfies (\ref{propertyB}).

Suppose now that we have an $[n^{(k)},k]_q$ MWS code $\C_k$ satisfying (\ref{propertyB}). We want to show that we can build an $[n^{(k+1)},k+1]_q$ code $\C_{k+1}$ that satisfies the same property. By Theorem \ref{thm:cr}, we can find $r \in \N^{q-1}$ such that  $\C_k(r)$ is an $[Rn,k]_q$ MWS code satisfying (\ref{propertyA}) and  (\ref{propertyB}). Therefore, by part 2 of Proposition \ref{prop:recursive} we get an $[n^{(k+1)},k+1]_q$ MWS code $\C_{k+1}$ with property  (\ref{propertyB}).
\end{proof}

\begin{example}\label{exa1}
Here we see what happens in the easiest case that was not covered in \cite{shi18}, i.e. when $q=k=3$. We consider the finite field $\F_3=\{0,1,2\}$. The code $\C_1 \subseteq \F_3^3$ is $\{(0,0,0), (1,2,2)\}$ and the inductive construction gives a $[7,2]_3$ code $\C_2$ whose generator matrix is
$$G_2=\begin{pmatrix}
 1 & 2 & 2 & 0 & 0 & 0 & 0  \\
 1 & 1 & 1 & 2 & 2 & 2 & 2 
\end{pmatrix}.$$
Here, we can first compute the vectors $V(c)$ that are
$$\{(1,2,4), (2,1,4), (3,4,0), (4,3,0), (0,5,2), (5,0,2), (6,0,1), (0,6,1) \},$$ and we compute the hyperplanes $H_i^{a,b}$ that are the orthogonal complements of the vectors
\begin{align*}
\{&(1,-1), (1,1), (1,3), (3,1), (1,-3), (1,-2), (2,-1), (3,-1), (5,-2), (4,-1), (3,-4), \\
  & (2,-3), (6,-5), (1,0), (1,-4), (2,-5), (3,-2), (4,-3), (0,1), (5-6) \}.
  \end{align*}
Furthermore, the hyperplanes $H_{i,\ell}^c$ are the orthogonal complements of 
$$\{(3,2), (2,3), (3,4), (4,3), (2,-3), (3,-2), (5,-1), (1,-5)\}.$$
At this point we can choose $r=(1,6)$ that is not contained in any of those hyperplanes, and consider the code $\C_2(r)$. Such a code is a $[49,2]_3$ code over $\F_3$ whose generator matrix is given by 
%$$\left(\begin{array}{ccccccccccccccccccccccccccc}
% 2 & 1 & 1 & 0 & 0 & 0 & 0 & 0 & 0 & 1 & 2 & 2 & 0 & 0 & 0 & 0 & 0 & 0 & 1 & 2 & 2 & 0 & 0 & 0 & 0 & 0 & 0  \\
% 2 & 2 & 2 & 1 & 1 & 1 & 1 & 1 & 1 & 1 & 1 & 1 & 2 & 2 & 2 & 2 & 2 & 2 & 1 & 1 & 1 & 2 & 2 & 2 & 2 & 2 & 2 
%\end{array}\right).$$
$$\left(2G_2 \mid G_2 \mid G_2 \mid G_2 \mid G_2 \mid G_2 \mid G_2 \right)$$

Here the  vectors $V(c)$, for $0\neq c\in \C_2(r)$, are given by
$$ \{(8,13,28), (13,8,28), (22,27,0), (27,22,0), (5,30,14), (30,5,14), (36,6,7), (6,36,7) \},$$ 
The code $\C_2(r)$ satisfies (\ref{propertyB}). Moreover, $a[1]\neq b[1]$ for all $a,b \in \C_2(r)$ with $a\neq b$. Hence it satisfies also (\ref{propertyA}) and we can use part 1 of Proposition \ref{prop:recursive} and then construct the $[99,3]_q$ code as described in the proof of that result. We first embed the code $\C_2(r)$ into $\F_3^{99}$, by concatenating every codeword to the $0$ vector of length $50$. Finally we add to the resulting code the codeword whose entries are all equal to $1$, and we consider the subspace generated. The obtained code $\C_3$ will have $13$ distinct non-zero weights, given by $$w(\C_3)=\{21,49,35,42,99,91,86,77,72,94,69,63,93\},$$
i.e. it is a $[99,3]_q$ MWS code.
%The vectors $V(c)$ are given by
%\begin{align*}
%\{ &(5,4,72), (4,5,72), (15,12,54), (12,15,54), (14,7,60), (7,14,70). (8,16,67), (16,8,67) \\
%   &(27.54.0).  (18,59,4), (18,58,5), (0,69,12), (0,66,15), (6,68,7), (6,61,14), (3,62,16), (3,70,8), \\
%   &(54,27,0), (59,18,4), (58,18,5), (69,0,12), (66,0,15), (68,6,7), (61,6,14), (62,3,16), (70,3,8) \} 
%\end{align*}
%and each of them has pairwise distinct entries.
\end{example}

\section{Length of  MWS codes}
In this section we investigate lower bounds on the lengths of  $[n,k]_q$ MWS codes. A trivial lower bound is of course given by
$$n\geq \theta_q(k-1).$$

%\begin{remark}
%Note that the condition in part \ref{prop:recursive:pt1} of Proposition \ref{prop:recursive} (and part \ref{prop:recursive:pt2} (a) implies that $ n\ge q^k $
%\end{remark}

In the case $q=2$ this bound is sharp, in the sense that it is possible to construct a $[\theta_2(k-1), k]_2$ MWS code for each $ k\ge 1 $ \cite[Theorem 1]{shi18}.  This lower bound is not optimal  when $q \geq 3$. Indeed, we have the following result.

\begin{lemma}\label{lem:geomLBlength}
Let $ \C $ be an $[n,k]_q $ MWS code with $ k\ge 2 $. Then 
\[ n\ge  \left\lceil\frac{q\cdot\theta_q(k-1)}{2}\right\rceil= \left\lceil\frac{q^{k+1}-q}{2(q-1)}\right\rceil= \left\lceil\frac{1}{2}\left[q^k+q^{k-1}+\cdots+q  \right]\right\rceil.\]
\end{lemma}

\begin{proof}
Let $ \C $ be an $ [n,k]_q $ MWS code with $ k\ge 2 $. Let the columns of a generator matrix correspond to the  $ [n,k,d]_q $ projective system $ \M $  in $ \Pi=PG(k-1,q)$.  Consider the set $ S $ of incident point-hyperplane pairs $ (P,\Lambda) $, where $ P\in \M $. Summing over all members of $ \M $ we obtain
\begin{equation}\label{eqn:lb-1a}
|S| = \sum_{P\in \M} \theta_q(k-2) = n \cdot \theta_q(k-2).
\end{equation}
On the other hand, summing over all hyperplanes of $ \Pi $ we obtain

\begin{equation}\label{eqn:lb-1b}
|S| =\sum_{H\in \Pi} \ch_\M(H) \ge \sum_{i=0}^{\theta_q(k-1)-1}i = \frac{\theta_q(k-1)(\theta_q(k-1)-1)}{2}=\frac{\theta_q(k-1)\cdot q\cdot\theta_q(k-2)}{2},
\end{equation}
where the inequality is a consequence of Lemma \ref{lem:GeometricConditions}.
The result follows from  (\ref{eqn:lb-1a}) and (\ref{eqn:lb-1b}).
\end{proof}

As we have seen in Theorem \ref{thm:LongGeomExamples}, there exist $[n,k]_q$ MWS codes of length $n=\bo(2^{q^{k-1}})$. We are therefore motivated to determine values of $ n $ for which $[n,k]_q$ MWS codes exist with 
$$ \left\lceil\frac{q\cdot\theta_q(k-1)}{2}\right\rceil \le n < 2^{\theta_q(k-1)}-1.$$

\begin{corollary}
If $ \C $ is an  $ [n,k]_q $ MWS code satisfying property (\ref{propertyA}), then 
\begin{equation}\label{eqn:MWSpropertyA}
n\ge \left\lceil\frac{q^{k+2}-3q+2}{4(q-1)}\right\rceil. 
\end{equation}

\end{corollary}
\begin{proof}
By Proposition \ref{prop:recursive}, $\C$ gives rise to an $ [2n+1,k+1]_q $ MWS code $ \bar{\C} $. The result follows by applying the bound in Lemma \ref{lem:geomLBlength} to $ \bar{\C} $
\end{proof}

\begin{remark}
 Observe that in the case $q=2$, the bound (\ref{eqn:MWSpropertyA}) becomes 
$$n\ge \left\lceil\frac{2^{k+2}-6+2}{4}\right\rceil=2^{k}-1=\theta_2(k-1),$$
which coincides with the bound in Lemma \ref{lem:geomLBlength}. For $ q\ge 3 $ the two bounds diverge, so the question remains as to whether the bound in Lemma \ref{lem:geomLBlength} may be sharp for some $q\ge 3$.  For $k=2$ we have an answer.

%For the case $q\geq 3$ we have 
%$$n\ge \left\lceil\frac{q^{k+2}-3q+2}{4(q-1)}\right\rceil =\left\lceil\frac{q^{k+2}-2q^{k+1}-q+2}{4(q-1)}+\frac{2q^{k+1}-2q}{4(q-1)}\right\rceil. $$
%Since for $k\geq 1, q \geq 3$, we have that $\frac{q^{k+2}-2q^{k+1}-q+2}{4(q-1)}\geq 1$, we obtain that the two bounds are different. This means that for $q\geq 3$, an $[n,k]_q$ MWS code that satisfies property (\ref{propertyA}) can not have length equal to $\left\lceil\frac{q^{k+1}-q}{2(q-1)}\right\rceil. $ Therefore, at least one of the following is true: the bound in Lemma \ref{lem:geomLBlength} is not sharp for $q\ge 3$, or; the algebraic construction never gives  an $[n,k]_q$ MWS code with minimum length possible when $q\geq 3$. However, for $k=2$ we have an answer.
\end{remark}

\begin{proposition}\label{Prop: sharp for k=2}
For $ k=2 $, the bound in Lemma \ref{lem:geomLBlength} is sharp for all prime powers $ q $.
\end{proposition}

\begin{proof}
Denote the points of $ \ell=PG(1,q)=\{P_0,P_1,\ldots,P_q\} $, and consider the projective system $ \M $ in which every point $P_i$ appears with multiplicity $i$. The corresponding linear code is MWS and of length $ n= \frac{q(q+1)}{2}$.
\end{proof}

%\begin{remark}
%Note that, for $q\geq 3$, the code in the proof of Proposition \ref{Prop: sharp for k=2} cannot have  property (\ref{propertyA}).
%\end{remark}

\begin{proposition}\label{prop:projSmallq}
 There exists an $[7,3]_2$ MWS code, and there exists an $ [32,3]_3 $-MWS code.
\end{proposition}
\begin{proof}
For the first part, pick a triangle of points $ P_0,P_1,P_2$ in the (Fano) plane $ \Pi= PG(2,2) $. Construct a projective system $ \M $, whereby $ m(P_i)=2^i $.  One easily verifies that the characters of the seven lines in $ \Pi $ are $ 0,1,2,\ldots,6 $ respectively. Hence, the corresponding code is MWS.\\ 
For the second part, let $ \Pi=PG(2,3) $, and define the projective system $ \M $ with point multiplicities as indicated in the diagram.

\begin{figure}[h]
\centering
\begin{tikzpicture}
\node[circle, draw,inner sep=0,minimum size=3cm,outer sep=0, thick] (c) at(0,0){};
\node[regular polygon, regular polygon sides=3, draw,
inner sep=0,minimum size=6cm, thick] (t) at (0.0,0) {};

\foreach \nn in {c.270,t.90,t.210}{
\node[fill=black,circle] at(\nn){};
}

\node[fill=black,circle] at (0,0){};
\draw[thick] (c.30) -- (t.210) (c.150) -- (t.330)  (c.270) -- (t.90);

\node[fill=blue,circle, text=white, inner sep=0, minimum size=14] at (c.30) {2};
\node[fill=blue,circle, text=white, inner sep=0, minimum size=14] at (c.150) {1};
\node[fill=blue,circle, text=white, inner sep=0, minimum size=14] at (t.330) {4};
%\node[fill=blue,circle, text=white, inner sep=0, minimum size=14] at (t.210) {0};
\end{tikzpicture}
\begin{tikzpicture}[
every node/.style={fill=blue,circle, text=white, inner sep=0, minimum size=14},
every path/.style={draw=white,double=black, very thick}
]

\path (-1,-1)  node (7) {8} 
  (-1,0)  node (4) {1}
  (-1,1)  node (1) {1}
  (0,-1)  node (8) {4}
  (0,0)  node[fill=black,circle, text=white, inner sep=0, minimum size=14] (5) { }
 (0,1)  node (2) {2}
 (1,-1)  node (9) {1}
 (1,0)  node[fill=black,circle, text=white, inner sep=0, minimum size=14] (6) { }
 (1,1)  node (3) {3};

\path (45:3) node[fill=black,circle, text=white, inner sep=0, minimum size=14] (10) { }  (0:3) node[fill=black,circle, text=white, inner sep=0, minimum size=14] (11) { }  (-45:3) node[fill=black,circle, text=white, inner sep=0, minimum size=14] (12) { }  (-90:3) node (13) {12};

\draw (1) to (2) to (3) to[out=0,in=135] (11);
\draw (4) to (5) to (6) to (11);
\draw (7) to (8) to (9) to[out=0,in=-135] (11);

\draw (1) to (4) to (7) to[out=-90,in=135] (13);
\draw (2) to (5) to (8) to (13);
\draw (3) to (6) to (9) to[out=-90,in=45] (13);

\draw (7) to[out=135,in=135, distance=100] (2) (2) to (6) (6) to[out=-45,in=90] (12);
\draw (1) to (5) to (9) to (12);
\draw (3) to[out=135,in=135, distance=100] (4) (4) to (8) (8) to[out=-45,in=180] (12);

\draw (1) to[out=-135,in=-135, distance=100] (8) (8) to (6) (6) to[out=45,in=-90] (10);
\draw (7) to (5) to (3) to (10);
\draw (9) to[out=-135,in=-135, distance=100] (4) (4) to (2) (2) to[out=45,in=180] (10);

\draw[bend left=17] (10) to (11) (11) to (12) (12) to (13);

\end{tikzpicture}

\caption{ } \label{fig:PG(2,3)}
\end{figure}

The characters of the lines of $ \Pi $ are $ 1,2,4,5,6,8,10,11, 12, 13, 16, 18, $ and $ 22 $ respectively.
\end{proof}
Note that the existence of a $ [7,3]_2 $-MWS code is also shown in \cite{shi18}.

\newpage

\begin{proposition}\label{prop:proj3D}
 For each $q$, there exists an $[n,3]_q$ MWS code with $n\le \frac{q-1}{2}(q^3+q^2+q)$.
\end{proposition}

\begin{proof}
For $ q=2,3 $, the result follows from Proposition \ref{prop:projSmallq}, so assume $ q>3 $.  Let $ \Pi=PG(2,q) $  and consider a set $ T=\{P,Q,R\} $ of three non-collinear points, and the three lines that are their joins, $ \ell_0=\langle P,R \rangle $, $ \ell_1=\langle P,Q\rangle  $, and $ \ell_2=\langle Q,R\rangle. $
Let $ \ell_0\setminus T=\{P_1,P_2,\ldots,P_{q-1}\} $, $ \ell_1\setminus T=\{Q_1,Q_2,\ldots,Q_{q-1}\} $, and $ \ell_2\setminus T=\{R_1,R_2,\ldots,R_{q-1}\}$. 

\begin{figure}[h]
\centering
 \begin{tikzpicture}[dot/.style={circle,inner sep=2pt,fill,label={#1},name=#1}, 
 every path/.style={draw=black, thick},
   extended line/.style={shorten >=-#1,shorten <=-#1}, 
   extended line/.default=1cm] 
    
 \node [fill=blue,circle, text=white, inner sep=0, minimum size=8, label=below:{$ P $},inner sep=1pt] (P) at (0,0) {};
 \node [fill=blue,circle, text=white, inner sep=0, minimum size=8,label=above:{$ R $},inner sep=1pt] (R) at (2,2) {};
 \node [fill=blue,circle, text=white, inner sep=0, minimum size=8,label=left:{$ Q $},inner sep=1pt] (Q) at (1,2.5) {};
 \node [label=right:{$ \ell_0 $},inner sep=1pt] (B) at (3,3) {};
 \node [label=right:{$ \ell_1 $},inner sep=1pt] (C) at (1.5,3.8) {};
 \node [label=above:{$\ell_2 $},inner sep=1pt] (M) at (-1,3.5) {};

 %\draw [extended line=0.0cm] (P) -- (B);
 % \draw [extended line=1.0cm] (R) -- (P);  
 \draw [extended line=1.5cm] (P) -- (R);
  \draw [extended line=1.5cm] (Q) -- (P);  
  \draw [extended line= 2.5cm] (R) -- (Q); 
 \end{tikzpicture} 
\caption{ } \label{fig:triangle}
\end{figure}

Next, we assign multiplicities to points to construct the projective system $ \M $. \\
Consider first the case that $ q>2 $ is even. 
 Let $ m(P_i)=i $, $ m(Q_i)=iq $, and $ m(R_i)=iq^2 $, $ 1\le i \le q-1 $. We claim that no two lines of $ \Pi $  have the same character.  For each line, $ \ell $,   $ \ch_{\M}{(\ell)} $  may be written uniquely as $ a+bq+cq^2+dq^3 $,  where $ 0\le a,b,c,d\le q-1 $. Moreover, $ |\ell\cap T|=t $, where $ 0\le t\le 2 $. Consider the following cases.
\begin{description}
\item[t=0:] In this case, $ 1\le a,b,c\le q-1 $, and $ d=0 $. Moreover, as two points uniquely determine a line, no two such lines agree in as many as two of $ a,b,c $. 
\item[t=1:] In this case, exactly one of $ a,b,c $ are non-zero, and $ d=0 $. That no two such lines have the same character follows from the fact that two points determine an unique line, and that no two (distinct) points in $ \{\ell_0\cup\ell_1\cup\ell_2\}\setminus T $ have the same multiplicity.
\item[t=2:] For $ \ell_0 $ we have $ a=\frac{q}{2} $, $ b=\frac{q-2}{2} $, $ c=d=0 $; for $ \ell_1 $ we have $ b=\frac{q}{2} $, $ c=\frac{q-2}{2} $, $ a=d=0 $; and for $ \ell_2 $ we have $ c=\frac{q}{2} $, $ d=\frac{q-2}{2} $, and $ a=b=0 $.           
\end{description}    
Consider now the case that $ q>2 $ is odd. In this case, we assign all multiplicities as in the case $ q $ is even, with the following exceptions: $ m(R_{\frac{q-1}{2}})=0 $, and $ m(P)=\frac{q-1}{2}q^2 $. For each line, $ \ell $,   $ \ch_{\M}{(\ell)} $  may be written uniquely as $ a+bq+cq^2+dq^3 $,  where $ 0\le a,b,c,d\le q-1 $. Moreover, $ |\ell\cap T|=t $, where $ 0\le t\le 2 $. Consider the following cases.
\begin{description}
\item[t=0:] In this case, $ 1\le a,b\le q-1 $, $ 0\le c \le q-1 $, $ c\ne \frac{q-1}{2} $, and $ d=0 $. As two points uniquely determine a line, no two such lines agree in both $ a$ and $b$. 
\item[t=1:] If $ P\in \ell $ then $ a=b=0, $ $ 0\le c\le q-1 $, and $ 0\le d \le 1 $. If $ Q\in \ell $ then $ b=c=0 $, $ 1\le a \le q-1 $, and $ d=0 $. If $ R\in \ell $ then $ a=c=0 $, $ 1\le b\le q-1 $, and $ d=0 $. No two such lines have the same character follows from the fact that two points determine an unique line, and that no two (distinct) points in $ \{\ell_0\cup\ell_1\cup\ell_2\}\setminus T $ have the same multiplicity.
\item[t=2:] If $\ell= \ell_0 $ then $ b=c=\frac{q-1}{2} $, and $ a=d=0 $. If $\ell= \ell_1 $ then $ a=b=d=0 $, and $ c=q-1 $. If $\ell= \ell_2 $ then $ a=b=0 $, $ c=\frac{q+1}{2} $, and $ d=\frac{q-3}{2} $.
\end{description} 
Consequently, for $ q>3 $ we arrive at an $[n,3]_q $ MWS code $ \C$, where \[ n=\sum_{j=0}^{2}\sum_{i=1}^{q-1}iq^{j} = \frac{q-1}{2}(q^3+q^2+q).\]
%This proves the first part.  To prove the second part of the proposition, consider the same projective systems as in part (a), but supplemented by setting $ \ch_{\M}(P)=\ch_{\M}(Q)=\ch_{\M}(R)=q^3$. The result follows by the same arguments as in part (a), observing that now, in each of the cases considered, $ t=d $.  

\end{proof}

\begin{lemma} \label{lem:PSspans}
If $ C $ is an $ [n,k]_q $-MWS code with projective system $ \M $ in $ \Pi= $PG$ (k-1,q) $, then $ \M $ spans $ \Pi $.
 \end{lemma}
\begin{proof}
The result clearly holds for $ k=2 $, where $ \M $ must contain at least $ q>2 $ distinct points. Consider $ k>2 $, and suppose $ \lambda $ is an hyperplane of $ \Pi $ with $ \M\subset \Pi$. Let $ \sigma $ be any $ (k-2) $-flat within $ \lambda $. Through $ \sigma $ there are $ q $ hyperplanes distinct from $ \lambda $, each of which has character $ \ch_{\M}(\sigma) $. Since $ q>2 $ we have a contradiction.  
\end{proof}

\begin{theorem}\label{thm:projectivefamily}
Suppose there exists an $ [n,k]_q $-MWS code with $ n<q^t $. Then there exists an $ [N,k+1]_q $-MWS code with $ N<q^{t+k+1} $.
\end{theorem}

\begin{proof}
Let $ \M $ be the projective system in $ \Pi=PG(k-1,q) $ associated with the $ [n,k]_q$-MWS code, $ C $. Let $ T=\{ T_0,T_1,\ldots,T_{k-1} \}$ be a set of $ k $ points of $ \Pi $ that are in general position ($ T $ spans $ \Pi $). Note that such a set may always be chosen. Embed $ \Pi $ in $ \Sigma=PG(k+1,q) $, and distinguish a point $ P \in \Sigma^*=\Sigma\setminus\Pi$.\\

\begin{figure}[h]
\centering
\begin{tikzpicture}
%\draw[help lines] (0,0) grid (12,6);

%\node[draw,,thick,minimum width=12cm,minimum height=3.5cm, thick] (c) at(6,1.6){};

\node[trapezium, draw, minimum width=10cm, minimum height=3.5cm,
trapezium left angle=60, trapezium right angle=60] at(6,1.6){};

\node[circle, inner sep=0,minimum size=4cm,outer sep=0,] (c) at(6,2.5){};

\node[fill=blue,circle, text=white, inner sep=0, minimum size=9, label=left:{$ T_0 $}] (T0) at (c.180) {};
\node[fill=blue,circle, text=white, inner sep=0, minimum size=9, label=south west:{$ T_1 $}] (T1) at (c.220) {};
\node[fill=blue,circle, text=white, inner sep=0, minimum size=9, label=south:{$ T_2 $}] (T2) at (c.270) {};
\node[fill=blue,circle, text=white, inner sep=0, minimum size=9, label=south east:{$ T_3 $}] (T3) at (c.300) {};
%\node[fill=blue,circle, text=white, inner sep=0, minimum size=9, label=left:{$ T_0 $}] at (c.300) {2};
\node[fill=blue,circle, text=white, inner sep=0, minimum size=9, label=east:{$ T_{k-1} $}] (Tk) at (c.360) {};

\node[fill=blue,circle, text=white, inner sep=0, minimum size=9, label=west:{$ P $}] (P) at (6,6) {};

\foreach \i in {c.320, c.330, c.340}{
\node[fill=blue,circle, text=white, inner sep=0, minimum size=3] at(\i){};
}

\draw[thick] (P) to (T0);
\draw[thick] (P) to (T1);
\draw[thick] (P) to (T2);
\draw[thick] (P) to (T3);
\draw[thick] (P) to (Tk);

\node[label=west:{$ \Sigma =PG(k,q) $}] (P) at (4,5) {};
\node[label=west:{$ \Pi =PG(k-1,q) $}] (P) at (4.7,-.3) {};
\end{tikzpicture}
\caption{ } \label{fig:ProjInduction}
\end{figure}

For each $ i $,  $ 0\le i \le k-1  $, let $ \ell_i=\langle P,T_i\rangle =\{P,T_i,Q_{i,0},Q_{i,1},\ldots,Q_{i,q-1}\} $. Define the projective system $ \M' $ by letting $ \ch_{M'}(R)=\ch_M(R) $ for each point $ R\in \Pi $, and by letting $ \ch_{\M'}(Q_{i,j})= jq^{t+1} $. We claim that $ \M' $ is a projective system of an MWS code.\\
Suppose $ \lambda$ is an hyperplane, with $ \ch_{\M'}(\lambda)= \ch_{\M'}(\Pi)$. Each hyperplane must meet each $ \ell_i $ in at least one point. For each $ i,j $, we have $ m(Q_{i,j})>n= \ch_{\M'}(\Pi) $.  If $ \lambda\ne \Pi $, then it must be the case that $ \lambda\cap \ell_i=\{P\} $ for $ i=0,1,\ldots,k-1 $, and  therefore $ \ch_{\M'}(\lambda)= \ch_{\M}(\lambda)=n $, contradicting Lemma \ref{lem:PSspans}. It follows that if $ \lambda\ne \Pi $ is an hyperplane of $ \Sigma $, then $ \ch_{\M'}(\lambda)\ne \ch_{\M'}(\Pi) $. It therefore suffices to consider hyperplanes distinct from $ \Pi $.\\
Let $ \lambda_1,\lambda_2 $ be hyperplanes with $ \ch_{\M'}(\lambda_1)=a_0+a_1q+a_2q^2+\cdots+a_{k+t}q^{k+t}$, and $ \ch_{\M'}(\lambda_2)=b_0+b_1q+b_2q^2+\cdots+b_{k+t}q^{k+t}$. If $ \lambda_1\cap \Pi \ne \lambda_2\cap\Pi $, then $ (a_0,a_1,a_2,\ldots,a_{t-1})\ne (b_0,b_1,b_2,\ldots,b_{t-1}) $  ( $ \ch_{\M'}(\lambda_1\cap\Pi)\ne \ch_{\M'}(\lambda_2\cap\Pi)  $ since $ C $ is MWS). It therefore suffices to consider the case that $ \lambda_1\cap\Pi =\lambda_2\cap\Pi $.\\
Let $ \sigma $ be a $ (k-2) $-flat in $ \Pi $. Since the points of $ T $ are in general position, there exists at least one $ j $ with  $ T_j\notin \sigma $. A dimension argument shows the hyperplanes containing $ \sigma $ must meet $ \ell_j $ in mutually distinct points. By considering $ \sigma=\lambda_1 \cap\lambda_2$, it follows that $ b_j\ne a_j $, and whence $ \ch_{\M'}(\lambda_1)\ne \ch_{\M'}(\lambda_2)$.\\
Therefore, $ \M' $ is a projective system of an $ [n,k+1]_q $-MWS code, where 
\[
N=n+\sum_{j=0}^{k-1}\sum_{i=1}^{q-1}iq^{t+j}=n+\frac{q-1}{2}\left[ q^{t+1}+q^{t+2}+\cdots +q^{t+k}\right]<q^{t+k+1}.
\]   
\end{proof}

\begin{corollary}\label{cor:projectiveExistence}
For each $ k\ge 3 $, there exists an $ [n,k]_q $-MWS code with $ n<q^{\frac{k^2+k-4}{2}} $
\end{corollary}
\begin{proof}
For $ k=3 $ the result follows from Proposition  \ref{prop:proj3D}. From the Theorem \ref{thm:projectivefamily}, and  Proposition \ref{prop:proj3D} we have for $ k>3 $ there exists an $ [n,k]_q $-MWS code with   
\[
n<q^{4+(4+5+\cdots+k)}=q^{\frac{k(k+1)}{2}-2}
\]  
\end{proof}

\begin{remark}
The bound in Lemma \ref{lem:geomLBlength} is rather optimistic in the following sense: The only way a code can achieve this bound is if there is an hyperplane of every character, from $ 0 $ to $ \theta_q(k-1)-1 $.  For $ k>2 $, constructions for codes meeting the bounds in Lemma \ref{lem:geomLBlength}, even asymptotically, seem quite elusive.  	
\end{remark}

\subsection{Length of the codes arising from the algebraic construction}
In order to give a partial answer to the question above about the length of MWS codes, we will try to estimate the length $n^{(k)}$ of the code $\C_k$ obtained using the construction given in Section \ref{sec:alg}. We will show indeed that asymptotically the algebraic construction provides codes of considerably smaller length with respect to those given in Theorem \ref{thm:LongGeomExamples}. For this purpose we need to understand how the length increases after each of the two partial steps.

We will denote by $\C_k$ the $k$-dimensional MWS code, by $n^{(k)}$ its length and by $r^{(k)} \in \N^{q-1}$ the vector that minimize the quantity $R^{(k)}=r_1^{(k)}+\ldots+r_{q-1}^{(k)}$, obtained in Theorem \ref{thm:cr}. We want to find some recurrence relation or upper bound for $n^{(k)}$, and this certainly involves also $R^{(k)}$ and  $T^{(k)}=\max\{c^r[\beta] \mid c^r\in \C_k(r^{(k)})\setminus \{0\}, \beta \in \F_q^*\}$. Since we want to minimize the length of our MWS codes, at each step the vector $r^{(k)}$ is chosen as a vector that minimizes $R^{(k)}$ among all the vectors that satisfy both  (\ref{eq:inters}) and  (\ref{eq:2}).
%Concerning the step given by Proposition \ref{prop:recursive}, we can see that the construction proposed for $\C_{k+1}$ works also adding the vector 
%$$(\underbrace{1,\ldots,1}_{n^{(k)} \text{ times }},\underbrace{\alpha,\ldots,\alpha}_{n^{(k)}+1 \text{ times }},\underbrace{\alpha^2,\ldots,\alpha^2}_{n^{(k)}+1 +(T^{(k)}+1) \text{ times }},\ldots,\underbrace{\alpha^{q-2},\ldots,\alpha^{q-2}}_{n^{(k)}+1+(q-3)(T^{(k)}+1) \text{ times }})$$
%where $T^{(k)}:=\max\{c^e[\beta] \mid c^e\in \C_k(r^{(k)})\setminus \{0\}, \beta \in \F_q^*\}$.
%At this point the question transforms into finding an upper bound for $T^{(k)}$. 
%
%Let $r^{(k)}=(r_1^{(k)},\ldots,r_{q-1}^{(k)})$ be the vector found with the  step of the iterative construction given by Theorem \ref{thm:cr} and let $R^{(k)}=r_1^{(k)}+\ldots+r_{q-1}^{(k)}$.
% Let moreover $S^{(k)}=\max \{c[\beta] \mid c\in \C_k, \beta \in \F_q^*\}$. Then 
%$$ T^{(k)}\leq R^{(k)}S^{(k)}.$$
%For the code $\C_1$, we have $S^{(1)}=q-1$, while in general, 

%$$S^{(k+1)}=n^{(k)}+1+(q-3)T^{(k)}+T^{(k)}=n^{(k)}+1+(q-2)T^{(k)}$$
%i.e.
%\begin{equation}
%T^{(k+1)}\leq R^{(k+1)}\left(n^{(k)}+1+(q-2)T^{(k)}\right)
%\end{equation}
Observe that we start with $\C_1=\F_q$ of length $n^{(1)}=1$, $r^{(1)}=(2,3,\ldots,q-1, 1)$, $R^{(1)}=\frac{q(q-1)}{2}$ and $\C_1(r^{(1)})$ is the code generated by the codeword
$$c=(1,\alpha, \alpha, \alpha^2, \alpha^2, \alpha^2, \ldots, \underbrace{\alpha^{q-2},\ldots,\alpha^{q-2}}_{q-1 \text{ times }}). $$

%For the parameter $T^{(k)}$, we have $T^{(1)}=q-1$, and in general it satisfies (I have to check more carefully)

\begin{proposition}
The sequences of integers $n^{(k)}$ and $T^{(k)}$  satisfy the following recurrence relations.
\begin{equation}\label{eq:Tk}
T^{(k+1)}\leq \sum_{j=0}^{q-2}r_{i_j} (n^{(k)}+j (T^{(k)}+1))=R^{(k)}n^{(k)}+(T^{(k)}+1)\sum_{j=0}^{q-2}j r_{i_j},
\end{equation}
where we have reordered the $r_i$'s such that $r_{i_0} \leq r_{i_1} \leq \ldots \leq r_{i_{q-2}}$, and

\begin{equation}\label{eq:nk}
n^{(k+1)}=(q-1)R^{(k)}n^{(k)}+1+\frac{(q-3)(q-2)}{2}(T^{(k)}+1),
\end{equation}
with initial conditions $n^{(1)}=1$ and $T^{(1)}=q-1$.
\end{proposition}

\begin{proof}
By Theorem $\ref{thm:cr}$, $\C_k(r^{(k)})$ has length $R^{(k)}n^{(k)}$, and then by part 2 of Proposition \ref{prop:recursive} we obtain $n^{(k+1)}=(q-1)R^{(k)}n^{(k)}+1+\frac{(q-3)(q-2)}{2}(T^{(k)}+1)$.

For the parameter $T^{(k+1)}$ we consider what happens when we go from $\C_k(r^{(k)})$ to $\C_{k+1}(r^{(k+1)})$. Let $T^{(k)}=\max\{c^r[\beta] \mid c^r\in \C_k(r^{(k)})\setminus \{0\}, \beta \in \F_q^*\}$. After applying the construction of part 2 of Proposition \ref{prop:recursive}, we get a code $\C_{k+1}$ and now consider a codeword $c\in \C_{k+1}$. If we sort in increasing order the values $c[\alpha^i]$, we get
$$c[\alpha^{\ell_0}]<c[\alpha^{\ell_1}]<\ldots <c[\alpha^{\ell_{q-1}}],$$
and it is easy to check that for every $j=1,\ldots,q-1$, we have that 
\begin{equation}\label{eq:ca}
c[\alpha^{\ell_j}]\leq n+1+(j-1)(T^{(k)}+1)+ T^{(k)}.
\end{equation}
At this point we construct the code $\C_{k+1}(r^{(k+1)})$. In this code, by part 5 of Proposition \ref{prop:Cr}, the entries of $V(c^r)$ for $c^r\in \C_{k+1}(r^{(k+1)})$ are given by 
$$c^r[\alpha^i]=\sum_{j=1}^{q-1}r_jc[\alpha^{i-j}].$$
Reordering the $r_i$'s such that $r_{i_0} \leq r_{i_1} \leq \ldots \leq r_{i_{q-2}}$, we get
$$c^r[\alpha^i]=\sum_{j=1}^{q-1}r_jc[\alpha^{i-j}]\leq \sum_{j=1}^{q-1}r_{i_j}c[\alpha^{\ell_j}].$$
Combining this inequality with (\ref{eq:ca}), we get the desired result.
\end{proof}

Now it only remains to determine an upper bound for $R^{(k)}$.
One possible strategy to find an estimate on $R^{(k)}$ is the following and relies on Schwartz-Zippel Lemma \cite{sc80}.
Let $H$ be the union of $D^{(k)}$  distinct linear hyperplanes obtained after the $k$-th step of the algebraic construction. A translation of the  Schwartz Zippel Lemma \cite[Corollary 1]{sc80} states that
\begin{equation}\label{eq:SZ}
|\left\{x=(x_1,\ldots, x_{q-1}) \mid x \in H, x_i\in \{0,\ldots, m-1\}, x\neq 0\right\}|\leq D^{(k)} (m^{q-2}-1).
\end{equation}
We can also find an upper bound on $D^{(k)}$.
\begin{lemma}\label{lem:boundDk}
Let $D^{(k)}$ denote the number of distinct linear hyperplanes obtained at the $k$-th step of the iterative construction.
$$D^{(k)}\leq \binom{q^k}{2}+(q^k-1)(q-1) $$
\end{lemma}

\begin{proof}
 There are  $\binom{q^k}{2}$  hyperplanes of the form $H_i^{a,b}$ with $a,b \in \C_k$ and $a\neq b$. Moreover, for $1\leq i <\ell$, we have that the linear hyperplanes $H_{i,\ell}^c=H_i^{c,\alpha^{\ell-i}c}$. Therefore, we already counted them. Finally, the hyperplanes of the form $H_{0,\ell}^c$, for $1\leq \ell \leq q-1$, $c \in \C_k\setminus \{0\}$, are exactly $(q^k-1)(q-1)$ many.
\end{proof}

\begin{corollary}\label{cor:ORk}
For $k\ge 2$, there exists an $r^{(k)}\in \N^{q-1}$ satisfying Theorem \ref{thm:cr} such that $$R^{(k)}\leq (q-1)\binom{q^k}{2}+(q^k-1)(q-1)^2=\bo(q^{2k+1}).$$
\end{corollary}

\begin{proof}
 Equation (\ref{eq:SZ}) with $m=D^{(k)}$ implies that there exists a non-zero $r^{(k)}=(r_1^{(k)},\ldots,r_{q-1}^{(k)})\in \N^{q-1}$ that is not in the union of the $D^{(k)}$ hyperplanes, such that each of the $r_i^{(k)}$ belongs to $\{0,1,\ldots,D^{(k)}-1\}$. This implies that 
$$R^{(k)}\leq (q-1)D^{(k)}\leq(q-1)\binom{q^k}{2}+(q^k-1)(q-1)^2,$$
where the last inequality follows from Lemma \ref{lem:boundDk}.
\end{proof}
%
%\begin{lemma}
%\begin{enumerate}
%\item $$ |\{(x_1,\ldots, x_{q-1}) \in \N^{q-1} \mid x_1+\ldots+ x_{q-1}=s\}| = \binom{s+q-2}{s}=\binom{s+q-2}{q-2} $$
%\item $$ \sum_{s=0}^r\binom{s+q-2}{q-2}=\binom{r+q-1}{q-1} $$
%\end{enumerate}
%\end{lemma}
Now we estimate the asymptotic order of the parameters $n^{(k)}$ and $T^{(k)}$.

\begin{proposition}\label{prop:bigo}
 For $k\ge 1$ we have
$$n^{(k+1)}=\bo(q^{(k+1)(k+2)-3}) \quad \mbox{ and } \quad T^{(k+1)}=\bo(q^{k(k+3)}).$$
\end{proposition}

\begin{proof}
 We prove it by induction. Recall that $n^{(1)}=1$, $T^{(1)}=q-1$ and $R^{(1)}=\frac{q(q-1)}{2}$. For $k=1$, we can compute 
\begin{align*}
T^{(2)}&\leq   R^{(1)}n^{(1)}+(T^{(1)}+1)\sum_{j=0}^{q-2}j r_{i_j}\\
    & =\bo(q^2)+\bo(q^4)=\bo(q^4),
\end{align*}
\begin{align*}
n^{(2)}&=(q-1)R^{(1)}n^{(1)}+1+\frac{(q-3)(q-2)}{2}(T^{(1)}+1)\\
&=\bo(q^{3})+\bo(q^3)=\bo(q^3).
\end{align*}

Suppose that the result holds for $n^{(k)}$ and $T^{(k)}$. By (\ref{eq:Tk}) we have 
\begin{align*}
T^{(k+1)} &\leq R^{(k)}n^{(k)}+(T^{(k)}+1)\sum_{j=0}^{q-2}j r_{i_j} \\
&=\bo(q^{2k+1+ k(k+1)-3})+\bo(q^{(k-1)(k+2)+2k+2})\\
&=\bo(q^{k(k+3)-2})+\bo(q^{k(k+3)})=\bo(q^{k(k+3)}).
\end{align*}
Moreover, by (\ref{eq:nk}), it holds 
\begin{align*}
n^{(k+1)}&=(q-1)R^{(k)}n^{(k)}+1+\frac{(q-3)(q-2)}{2}(T^{(k)}+1)\\
&=\bo(q^{2k+2+k(k+1)-3})+\bo(q^{(k-1)(k+2)+2})\\
&=\bo(q^{(k+1)(k+2)-3}).
\end{align*}
\end{proof}

At this point we can conclude our estimate of the length of MWS codes obtained via the algebraic construction. 

\begin{theorem}\label{thm:OMWS}
Using the construction given in Section \ref{sec:alg} we can obtain an $[n,k]_q$  MWS codes with length $n=\bo(q^{k(k+1)-4})$.
\end{theorem}

\begin{proof}
Using the proof of Theorem \ref{thm:algexist} we construct an $[n^{(k-1)}, k-1]_q$ MWS code $\C_{k-1}$ satisfying (\ref{propertyB}), with $n^{(k-1)}=\bo(q^{(k-1)k-3})$ by Proposition \ref{prop:bigo}. By Theorem \ref{thm:cr} there exists $r^{(k-1)} \in \N^{q-1}$ such that $\C_{k-1}(r^{(k-1)})$ is an $[R^{(k-1)}n^{(k-1)}]_q$ MWS code satisfying (\ref{propertyA}) and (\ref{propertyB}). Therefore, by part 1 of Proposition \ref{prop:recursive}, there exists a $[2R^{(k-1)}n^{(k-1)}+1,k]_q$ MWS code. Since by Corollary \ref{cor:ORk} $R^{(k-1)}=\bo(q^{2k-1})$, we finally obtain
$$2R^{(k-1)}n^{(k-1)}+1=\bo(q^{2k-1+k(k-1)-3})=\bo(q^{k(k+1)-4}). $$
\end{proof}

\begin{remark}
Observe that the asymptotic estimate of the length of an MWS given in Theorem \ref{thm:OMWS} is done considering the worst-case scenario, and therefore the algebraic construction could give a shorter code in practice. Indeed it could be possible to improve the asymptotic estimate of $R^{(k)}$ that we computed with a worst-case argument. 
\end{remark}

\section{Conclusions and open questions}
In this paper we have introduced the concept of $ [n,k]_q $ maximum weight spectrum (MWS) codes.  We provided two different constructions for MWS codes, showing that they exist for all dimensions, and over every finite field. In Lemma \ref{lem:geomLBlength} we provide a lower bound on the length of MWS codes, which is shown to be sharp for $k=2$. The infinite families of $ [n,k]_q $-MWS codes provided here have lengths  $n\ge \bo(q^\frac{k^2+k-4}{2})$.  This raises the natural question: How short may an  MWS code be?  In other words, the general problem for fixed $ k $, and $ q $, is to determine the least value $ n $ such that an $ [n,k]_q $-MWS code exists. It is with some trepidation that we refer to such codes as optimal MWS codes (the term optimal is over used, but seems most appropriate here).  Rather than pose a conjecture on the length of optimal MWS codes, the authors invite investigation into the existence of infinite families of $ [n,k]_q $-MWS codes with $\bo(q^k)\leq n \leq \bo(q^{\frac{k^2+k-4}{2}})$. 
\section*{Acknowledgement} 
The authors would like to thank Patrick Sol\'e for his helpful communications, and in particular for introducing the authors of this manuscript to each other.

\bibliographystyle{abbrv}
\bibliography{MWS}

\end{document}